\documentclass[a4paper,12pt]{article}
\usepackage{graphicx}
\usepackage{jheppub} 
\usepackage{braket}
\usepackage{array}
\usepackage{tcolorbox}
\usepackage{qcircuit}
\usepackage{amsthm,amsmath,amssymb}
\usepackage[toc]{glossaries}
\usepackage{subcaption}
\usepackage{cancel}

\newtheorem{theorem}{Theorem}
\numberwithin{theorem}{section}

\newtheorem{definition}[theorem]{Definition}

\newtheorem{lemma}[theorem]{Lemma}

\newtheorem{example}[theorem]{Example}
\let\oldexample\example
\renewcommand{\example}{\oldexample\normalfont}

\newtheorem{atomicexample}{Example}
\let\oldatomicexample\atomicexample
\renewcommand{\atomicexample}{\oldatomicexample\normalfont}

\usepackage[backend=biber, natbib=false, bibstyle=alphabetic,citestyle=alphabetic, doi=true, sorting=nyt, giveninits]{biblatex}

\title{\boldmath Uniqueness of Complementary Recovery in Holographic Error-Correcting Codes}

\author[a,b]{Julia Jones}
\author[a,c]{and Jason Pollack}

\affiliation[a]{Institute for Quantum \& Information Sciences, Syracuse University, NY 13210, USA}

\affiliation[b]{Department of Physics, Syracuse University, NY 13210, USA}

\affiliation[c]{Department of Electrical Engineering and Computer Science, Syracuse University, NY 13210, USA}
        
\emailAdd{jjone112@syr.edu}
\emailAdd{japollac@syr.edu}

\newcommand{\M}{\mathcal{M}}
\newcommand{\E}{\mathcal{E}}
\renewcommand{\H}{\mathcal{H}}

\abstract{
Holographic codes are a type of error-correcting code with extra geometric structure ensured by a ``complementary recovery'' property: given a division of the physical Hilbert space $\H$ into $\H_A$ and $\H_{\bar A}$, and an algebra of physical operators $\M\subseteq (\mathcal{L}(\H_A)\otimes I_{\H_{\bar A}})$, the logical operators in $\mathcal{L}(\H_L)\simeq \mathcal{L}(P\H)$ which can be created by acting in $\M$ are identical to the logical operators whose expectation values cannot be altered by acting in the commutant $\M^\prime$, and vice versa.
In \cite{Pollack_2022}, a uniqueness theorem was stated: the only possible tuple of (code, bipartition, algebra) which can exhibit complementary recovery is the maximal one $\M=P(\mathcal{L}(\H_A)\otimes I_{\H_{\bar A}})P$. We point out a counterexample to this result, using a ``non-adjacent'' bipartition of a four-qubit code proposed in \cite{Pollack_2022}. We show that the failure of uniqueness is due to a failure to enforce error correction against erasure of $\H_{\bar A}$, which requires enforcing the algebraic Knill-Laflamme condition $[P E_i^\dagger E_j P,\M]=0$ for each pair of error operators. When we add the additional requirement that $\M$ be correctable with respect to this channel, uniqueness is restored, and we re-prove the theorem of \cite{Pollack_2022} with this added assumption. We present the list of bipartitions of the ``atomic'' holographic codes in \cite{Pollack_2022} in which the correctability assumption can be violated.
}

\begin{document} 
\maketitle

\section{Introduction}

The AdS/CFT correspondence \cite{maldacena1999large,Gubser:1998bc,Witten:1998qj,Aharony:1999ti} is a proposed theory of holographic quantum gravity. Researchers have made a connection between the holographic properties of this correspondence and quantum error correction \cite{verlinde2013black, almheiri2015bulk, mintun2015bulk}. In particular, Harlow \cite{Harlow_2017} identified the criteria that makes a quantum error-correcting code holographic using the operator-algebra error correction language established in \cite{beny2007generalization,B_ny_2007}. 

The paper that our work follows from, \cite{Pollack_2022}, built on the understanding of holographic properties in terms of quantum error correction and gave multiple examples of simple explicit holographic codes. The paper stated a theorem (\cite{Pollack_2022}, Theorem 4.4): if the set of logical operations on a code reachable by acting on a factor of the physical Hilbert space is a von Neumann algebra, then the code is holographic and the algebra is the unique von Neumann algebra satisfying complementary recovery for this choice of code and bipartition. 

While calculating the von Neumann algebras for the example codes given in \cite{Pollack_2022}, we found a contradiction to this theorem. The purpose of this paper is to explain the contradiction and prove a new theorem which is not contradicted. We show that a lemma (\cite{Pollack_2022}, Lemma 4.3) used to prove this theorem is contradicted, and the contradiction ultimately traces back to the fact that the definitions given in \cite{Pollack_2022} for a von Neumann algebra being correctable and private with respect to a code are too weak. However, we re-prove the theorem under the additional assumption that the code satisfies a more standard notion of correctability.

\paragraph{Overview of Results} We'll now describe our results more technically in the context of operator-algebraic quantum error correction (AQEC) \cite{beny2007generalization,B_ny_2007,kribs2019quantumcomplementarityoperatorstructures}. In AQEC, we work with algebras of operators that act on the logical Hilbert space $\H_L$ and are protected from errors on the physical Hilbert space $\H$. The logical information in $\H_L$ is encoded into $\H$ by an isometry $V: \H_L \rightarrow \H$ and the physical Hilbert space is factorized into two subregions, $\H = \H_A \otimes \H_{\bar A}$ where $\bar A$ is the subregion on which the error has occurred. For our codes, we study the error channel that erases $\bar A$, i.e. takes a state on $\H_A \otimes \H_{\bar A}$ to $\rho_A \otimes I_{\bar A}/|\H_{\bar A}|$. 

The operators on $\H_L$ that are \emph{correctable} from the subregion $A$ are ones that can be ``seen'' given access only to $A$, and thus are protected from the erasure of $\bar A$. Operators on $\H_L$ that are \emph{private} from $A$ are ones that cannot be seen given only access to $A$. These ideas are made more precise by \cite{kribs2019quantumcomplementarityoperatorstructures} where the error on $\bar A$ is given by the action of a quantum channel $\E$. The definitions of correctable and private algebras given in \cite{Pollack_2022} were a generalization of the original definitions in terms of $\E$. We identify the error in this generalization and use the definitions of \cite{kribs2019quantumcomplementarityoperatorstructures} to modify and re-prove the existence and uniqueness theorem for codes with complementary recovery, subject to an additional correctability condition:

\begin{theorem}
  Say $P$ is a projection and $A$ is a subregion. Let $\M:=P(\mathcal{L}(\H_A) \otimes I_{\bar A})P$ be the image of operators projected onto $P\H$. If $\M$ is a von Neumann algebra (in particular, it is closed under multiplication) and is correctable from $\E$,
then it is the unique von Neumann algebra satisfying complementary recovery with $P$ and $A$. If it is not, then no von Neumann algebra satisfying complementary recovery exists.
\end{theorem}

This theorem relies on a strengthened definition of complementary recovery compared to \cite{Pollack_2022}, which we'll argue contains the correct criteria:

\begin{definition}
    
  If $P$ is a projection, $A$ is a subregion inducing the factorization $\mathcal{H} = \mathcal{H}_A \otimes\mathcal{H}_{\bar{A}} $, and $\M$ is a von Neumann algebra, then $(P, A, \M)$ exhibit complementary recovery if 
  
\begin{itemize}

    \item $\M \text{ is correctable from } \mathcal{E} \text{ with respect to } P$ where $\E: \mathcal{L(H)} \rightarrow \mathcal{L(H)}$ is the channel that erases $\bar{A}$ and tensors on $I_{\bar{A}}$,

       \item $\M' \text{ is correctable from } \bar{\mathcal{E}} \text{ with respect to } P$ where $\bar{\E}: \mathcal{L(H)} \rightarrow \mathcal{L(H)}$ is the channel that erases $A$ and tensors on $I_{A}$.
    
\end{itemize}
\end{definition}

\paragraph{Organization of the Paper} This paper is organized as follows. In Section~\ref{sec:background}, we give a brief background on von Neumann algebras and identify some of the theorems that will be used later on. We will then introduce the original operator-algebra error correction definitions and define our specific erasure channel $\E$. In Section~\ref{sec:cr}, we start from the incorrect definitions used in \cite{Pollack_2022} and exhibit the contradictions they give rise to. We then demonstrate why we need a better understanding of the erasure channel $\E$ to resolve the contradiction. We present and re-prove the existence and uniqueness theorem. Finally, in Section~\ref{sec:examples}, we give the von Neumann algebras for other correctable example codes in \cite{Pollack_2022} and their different bipartitions. We conclude with a discussion in Section \ref{sec:discussion}. Mathematica code we used to check complementary recovery and compute properties of logical algebras can be found on Github: https://github.com/juliaj25/CR-for-Holographic-Codes.git.

\section{Background}
\label{sec:background}

For holographic codes, the algebra of operators that act on the logical Hilbert space are specifically von Neumann algebras. In this section we will give a brief overview of some of the important properties of von Neumann algebras. Additionally, to understand the errors that occur on the physical Hilbert space we will define a general quantum channel and then specify the quantum channel that erases $\bar A$ for our codes. Finally, we give the definitions of correctable and private algebras for a error correcting code from \cite{kribs2019quantumcomplementarityoperatorstructures}.

\subsection{Von Neumann algebras}

Throughout this subsection we will repeat the definitions of \cite{Pollack_2022}. More comprehensive reviews can be found, for example, in the Appendix of \cite{Harlow_2017} or the lecture notes \cite{jones2003neumann}.

A von Neumann algebra is an algebra of linear operators acting over a complex Hilbert space which has an identity element and is closed under the adjoint operation. More precisely,

\begin{definition}[von Neumann algebra]
Let $\mathcal{L}(\mathcal{H})$ be the set of linear operators over a finite-dimensional complex Hilbert space $\mathcal{H}$. A von Neumann algebra is a subset $\mathcal{M} \subseteq \mathcal{L}(\mathcal{H})$ which is closed under:
\begin{itemize}
    \item (addition) if $A, \, B \in \mathcal{M}$ then $A + B \in \mathcal{M}$;
    \item (multiplication) if $A, \, B \in \mathcal{M}$ then $AB \in \mathcal{M}$; 
    \item (scalar multiplication) if $A \in \mathcal{M}$  and $c \in \mathbb{C}$ then $cA \in \mathcal{M}$; 
    \item (complex conjugation) if $A \in \mathcal{M}$ then $A^\dagger \in \mathcal{M}$;
\end{itemize}
and for which there exists an element $I \in \mathcal{M}$ such that for every $A \in \mathcal{M}$ we have $IA = A$.
\end{definition}
In this paper we restrict ourselves to the finite-dimensional case, in which von Neumann algebras can be thought of as algebras of matrices. 

Another important set of operators is the set that contains all operators that commute with every element of $\M$. We call this the commutant.

\begin{definition}[commutant] Given a von Neumann algebra $\mathcal{M} \subseteq \mathcal{L}(\mathcal{H})$ the commutant is the set
\begin{equation}
    \mathcal{M}^{\prime} \equiv \left \{ B \in \mathcal{L}(\mathcal{H}) \mid \forall A \in \mathcal{M}: AB = BA \right \}.
\end{equation}
\end{definition}

We will keep the primed notation to indicate the commutant throughout this paper. The set of operators that commute with the commutant is the original algebra.

\begin{theorem}[bicommutant theorem]
For every von Neumann algebra $\mathcal{M}\subseteq \mathcal{L}(\mathcal{H})$ we have that
\begin{equation}
    \mathcal{M}^{\prime \prime} \equiv (\mathcal{M}^{\prime})^\prime = \mathcal{M}.
\end{equation}
\end{theorem}

The Wedderburn decomposition of a von Neumann algebra is given by the classification theorem. 

\begin{theorem}[classification theorem]
For every von Neumann algebra $\mathcal{M}$ on a finite-dimensional Hilbert space $\mathcal{H}$ there exists a block decomposition of the Hilbert space,
\begin{equation}
    \mathcal{H} = \left[ \oplus_{\alpha}  \left( \mathcal{H}_{A_\alpha} \otimes \mathcal{H}_{\bar{A}_\alpha}  \right) \right] \oplus \mathcal{H}_0,
\end{equation}
such that
\begin{align}
    \label{eq:Wedderburn} 
    \mathcal{M} = \left[ \oplus_{\alpha}  \left( \mathcal{L}(\mathcal{H}_{A_\alpha}) \otimes I_{\bar{A}_\alpha}  \right) \right] \oplus 0, \\
    \mathcal{M}^{\prime} = \left[  \oplus_{\alpha}  \left( I_{A_\alpha} \otimes \mathcal{L}(\mathcal{H}_{\bar{A}_\alpha} ) \right) \right] \oplus 0,
\end{align}
where $\mathcal{H}_0$ is the null space and $0$ is the zero operator on $\mathcal{H}_0$. 
\end{theorem}

For simplicity, we will no longer include the null space (zero operator) when we write the decomposition. The structure given by the classification theorem is what allows us to explicitly calculate von Neumann algebras for various codes. Since $\H=\H_A \otimes \H_{\bar A}$, the von Neumann algebra on the logical Hilbert space, $\H_L$, is  $\M=V^\dagger(\mathcal{L}(\H_A) \otimes I_{\bar A})V$ where $V$ is the isometry $V:\H_L \rightarrow \H$.

\subsection{Operator-algebraic quantum error correction}
 We briefly review the operator-algebraic generalization of quantum error correction of \cite{beny2007generalization,B_ny_2007,kribs2019quantumcomplementarityoperatorstructures}. To discuss correctability and privacy of an algebra for a code we need the notions of a quantum channel and its complementary channel.
 
\begin{definition}
    On a Hilbert space $\H$, a quantum channel is a completely positive trace-preserving map $\mathcal{E} : \mathcal{L}(\mathcal{H}) \rightarrow \mathcal{L}(\mathcal{H}) $. The Stinespring dilation theorem \cite{Stinespring1955} tells us there is an ``environment'' Hilbert space $\H_C$ (with $|\H_C| \leq |\H|^2$), a state $\ket{\psi_C} \in \H_C$, and a unitary $U$ on $\H \otimes \H_C$ such that, $\forall \rho \in \mathcal{L}(\H)$,
    \begin{align}
        \E(\rho) = \text{Tr}_{\H_C}\circ U(\rho \otimes \ket{\psi_C}\bra{\psi_C})U^*=\text{Tr}_C\circ V\rho V^*,
    \end{align}
where $Tr_C$ is the partial trace from $\mathcal{L}(\H \otimes \H_C) \text{ to } \mathcal{L}(\H)$ and $V: \H \otimes \H_C \rightarrow \H$ is the isometry defined by $V\ket{\psi}=U(\ket{\psi \otimes \psi_C})$.
\end{definition}

The complementary channel $\E^C: \mathcal{L}(\H) \rightarrow \mathcal{L}(\H_C)$ is defined as 
\begin{align}
    \E^C(\rho) \equiv Tr_\H V\rho V^*= \sum_{k,l} \text{Tr}(\rho E_k^*E_l) \ket{k} \bra{l},
\end{align}
where $\ket{k}$ and  $\ket{l}$ are the basis states of the environment space $\H_C$ and $E_k,E_l$ are the Kraus operators of $\E$.

The specific quantum channel that our logical operators are protected against is the channel that traces out the information in $\bar{A}$ and then tensors on the maximally-mixed state in $\bar{A}$. This is a discarding channel followed by an appending channel, with Kraus operators
  \begin{align} 
  E_{ij}=\frac{1}{\sqrt{|\mathcal{H}_{\bar{A}}|}} 
  (I_A \otimes \ket{i}_{\bar{A}})(I_A \otimes \bra{j}_{\bar{A}}) = \frac{1}{\sqrt{|\mathcal{H}_{\bar{A}}|}} 
  (I_A \otimes \ket{i} \bra{j}_{\bar{A}}),
  \end{align}
  where $\ket{i}$ and $\ket{j}$ are basis states of $\mathcal{H}_{\bar{A}}$. We will also need the dual channel, $\mathcal{E}^\dagger$, defined as $\text{Tr}(\mathcal{E}(\rho)X)=\text{Tr}(\rho \mathcal{E}^\dagger (X))$. Since 
  \begin{align}
      \text{Tr}((\sum_{k}E_k\rho E_k^\dagger)X)=\text{Tr}(\sum_{k}\rho E_k^\dagger X E_k)=\text{Tr}(\rho \mathcal{E}^\dagger (X)),
  \end{align}
then the Kraus operators of $\mathcal{E}^\dagger$ are
  \begin{align} 
  E^\dagger_{ij}=\frac{1}{\sqrt{|\mathcal{H}_{\bar{A}}|}} 
  (I_A \otimes \ket{j}_{\bar{A}})(I_A \otimes \bra{i}_{\bar{A}}) = \frac{1}{\sqrt{|\mathcal{H}_{\bar{A}}|}} 
  (I_A \otimes \ket{j} \bra{i}_{\bar{A}}).
  \end{align}

We will now define a correctable and private algebra in terms of $\E$ \cite{kribs2019quantumcomplementarityoperatorstructures}:

  \begin{definition}\label{def:corr}
Let $\H$ be a Hilbert space and let $P = VV^*$ be a projection on $\H$. Given a channel $\mathcal{E} : \mathcal{L}(\mathcal{H}) \rightarrow \mathcal{L}(\mathcal{H}) $ with Kraus operators $\{E_k\}$, an algebra $\M \subseteq \mathcal{L}(\mathcal{PH}) $ is 

\begin{itemize}
    \item \textbf{correctable} for $\E$ with respect to $P$ if and only if 
\begin{align}
    [PE^\dagger_k E_lP, X]=0 \text{     }\forall X\in \M, \forall k,l.\label{eq:correct_correctable}
\end{align}
\item \textbf{private} for $\E$ with respect to P if 
\begin{align}
  P (\mathcal{E}^\dagger (\mathcal{L}(\mathcal{H}))P \subseteq \M'= \{ X\in \mathcal{L}(\mathcal{PH}) | [X,A]=0, \forall A\in \M \}.
\end{align}
\end{itemize}

\end{definition}
Here we've given the ``testable condition'' \eqref{eq:correct_correctable} for correctability (\cite{kribs2019quantumcomplementarityoperatorstructures}, Theorem 2.3), which phrases it in terms of a set of Kraus operators of the error channel rather than more abstractly in terms of the channel alone. This condition can be seen as an algebraic generalization of the Knill-Laflamme condition for standard error correction, which requires only that the $PE^\dagger_k E_lP$ commute with the algebra rather than being proportional to $P$.

 Note that the logical Hilbert space may be thought of as $\H_L$ or as the code subspace $P\H\subset\H$. These spaces are isomorphic to one another ($V \mathcal{H}_L \sim P \mathcal{H}$) and $V$ will be used to go back and forth between them.
\section{Complementary recovery}
\label{sec:cr}

We will now introduce the existence and uniqueness theorem presented in \cite{Pollack_2022} and demonstrate the contradiction. We must first introduce the $\textbf{incorrect}$ definitions used in \cite{Pollack_2022} for correctable and private von Neumann algebras :

\begin{definition}\label{def:wrong} [\cite{Pollack_2022}, Definition 4.1] Say $V : \H_L \to \H$ is an encoding isometry $V$ for some quantum error-correcting code, and $A$ is a subregion of $\H$ inducing the factorization $\H = \H_A \otimes \H_{\bar A}$. A von Neumann algebra $\M \subseteq \mathcal{L}(\H_L)$ is said to be:
    \begin{itemize}
        \item \textbf{correctable} from $A$ with respect to $V$ if $\M \subseteq V^\dagger (\mathcal{L}(\H_A)\otimes I_{\bar A})V$. That is: for every $O_L \in \M$ there exists an $O_A \in \mathcal{L}(\H_A)$ such that $O_L = V^\dagger (O_A \otimes I_{\bar A}) V$.
        \item \textbf{private} from $A$ with respect to $V$ if $V^\dagger (\mathcal{L}(\H_A)\otimes I_{\bar A})V \subseteq  \M'$. That is: for every $O_A \in \mathcal{L}(\H_A)$ it is the case that $ V^\dagger (O_A \otimes I_{\bar A}) V$ commutes with every operator in $\M$.
    \end{itemize}
\end{definition}

Then, the $\textbf{incorrect}$ lemma given in \cite{Pollack_2022}:

\begin{lemma} \label{lemma:corrpriv} [\cite{Pollack_2022}, Lemma 4.3] \textbf{Correctable from $A$ $\leftrightarrow$ private from $\bar A$.} A von Neumann algebra $\M$ is correctable from $A$ with respect to $V$ if and only if $\M$ is private from $\bar A$ with respect to $V$. 
\end{lemma}

Central to the existence and uniqueness theorem is the idea of complementary recovery. Because an incorrect definition of correctable is used in \cite{Pollack_2022}, complementary recovery is also incorrectly defined. The $\textbf{incorrect}$ definition is as follows:
 
  \begin{definition}\label{def:wrongcom} [\cite{Pollack_2022}, Definition 4.2] A code with encoding isometry $V:\H_L\to \H$, a subregion of the physical Hilbert space $A$ and a von Neumann algebra $\M \subseteq \mathcal{L}(\H_L)$, together $(V,A,\M)$, exhibit \textbf{complementary recovery} if:
    \begin{itemize}
        \item $\M$ is correctable from $A$ with respect to $V$: $\M \subseteq V^\dagger (\mathcal{L}(\H_A)\otimes I_{\bar A})V$,
        \item $\M'$ is correctable from $\bar A$ with respect to $V$: $\M' \subseteq V^\dagger (I_A \otimes \mathcal{L}(\H_{\bar A}))V$.
    \end{itemize}
\end{definition}

Finally, we give the original theorem as it is stated in \cite{Pollack_2022}:

\begin{theorem}\label{thm:whatalgebra} [\cite{Pollack_2022}, Theorem 4.4] \textbf{Uniqueness of the von Neumann algebra.} Say $V$ is an encoding isometry and say $A$ is a subregion. Let $\M := V^\dagger (\mathcal{L}(\H_A)\otimes I_{\bar A})V$ be the image of operators on $\H_A$ projected onto $\H_L$. If $\M$ is a von Neumann algebra (that is, it is closed under multiplication), then it is the unique von Neumann algebra satisfying complementary recovery with $V$ and $A$. If it is not, then no von Neumann algebra satisfying complementary recovery exists.
\end{theorem}

We will show how the incorrect definition of correctable, and therefore complementary recovery, is what gives rise to the contradiction in this theorem. 

\subsection{The contradiction}

We will use example E of \cite{Pollack_2022}, shown in Figure~\ref{fig:e}, to demonstrate the contradiction.

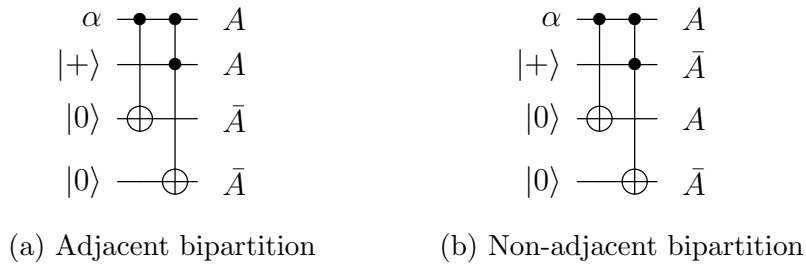
\begin{figure}[t]
\centering

\begin{subfigure}[b]{0.4\textwidth}
$$  \Qcircuit @C=0.3em @R=1.2em {
    \lstick{\alpha} & \ctrl{2} & \ctrl{1} & \qw & \rstick{A}\\
    \lstick{\ket{+}} & \qw &  \ctrl{2} & \qw  & \rstick{A}\\
    \lstick{\ket{0}} & \targ & \qw  & \qw & \rstick{\bar A}\\
    \lstick{\ket{0}}  & \qw & \targ & \qw &  \rstick{\bar A}
}$$
    \caption{Adjacent bipartition}
\end{subfigure}
\begin{subfigure}[b]{0.4\textwidth}
$$  \Qcircuit @C=0.3em @R=1.2em {
    \lstick{\alpha} & \ctrl{2} & \ctrl{1} & \qw &  \rstick{A}\\
    \lstick{\ket{+}} & \qw &  \ctrl{2} & \qw &  \rstick{\bar A} \\
    \lstick{\ket{0}} & \targ & \qw  & \qw &  \rstick{A}\\
    \lstick{\ket{0}}  & \qw & \targ & \qw &  \rstick{\bar A}
}$$
    \caption{Non-adjacent bipartition}
\end{subfigure}
\caption{Example code from \cite{Pollack_2022}}
\label{fig:e}
\end{figure}

We are given the isometry
\begin{align}
V=\frac{1}{\sqrt{2}}(\ket{0000}\bra{0}+\ket{0100}\bra{0}+\ket{1010}\bra{1}+\ket{1111}\bra{1}),
\end{align}
which takes the one qubit logical state $\alpha = a\ket{0}+b\ket{1}$ to a state in the 4 qubit physical space $\mathcal{H} = \mathcal{H}_A \otimes \mathcal{H}_{\bar{A}}$. There are two possible bipartitions of $\H$. The first is the adjacent bipartition where the first two qubits are in $\H_A$ and the second two qubits are in $\H_{\bar{A}}$. 

When calculating the algebra of operators on the logical state for this subregion $A_{1,2}$ using the definition $\mathcal{M} = V^{\dagger}(\mathcal{L}(\mathcal{H}_A)\otimes I_{\bar{A}})V$ from Theorem  \ref{thm:whatalgebra}, we find\footnote{Although some of the calculations that follow are simple enough to be performed by hand, we checked all statements about generators for algebras, etc., using Mathematica code which we have made available on Github:https://github.com/juliaj25/CR-for-Holographic-Codes.git} 
  \begin{align}
      \mathcal{M}_1= \text{span}\{Z, I\}.
  \end{align}
  Since $\M \subseteq \mathcal{L}(\H_L)$ is a subspace, we will write these algebras in terms of the operators that linearly span this subspace.
  
 The second choice of bipartition has nonadjacent subregions where the first and third qubit are in $\H_A$ and the second and fourth qubit are in $\H_{\bar{A}}.$ Calculating the algebra of operators for this subregion $A_{1,3}$, we find 
 \begin{align}
    \mathcal{M}_2= \text{span}\{X,Y,Z,I\}.
  \end{align}

The commutants of these algebras are 
\begin{align}
    \M'_1=\text{span}\{Z,I\}\\
    \M'_2=\text{span}\{I\}
\end{align}

$\M'_1$ is equal to $V^{\dagger}(I_{A}\otimes \mathcal{L}(\mathcal{H}_{\bar{A}}) )V$ for $\bar A_{3,4}$. However, $\M'_2$ is not equal to $V^{\dagger}(I_{A}\otimes \mathcal{L}(\mathcal{H}_{\bar{A}}) )V=\text{span}\{Z,I\}$ for $\bar A_{2,4}$. This is the first indication of a contradiction.  However, $\M_1$ and $\M_2$ are both von Neumann algebras and according to Theorem~\ref{thm:whatalgebra} complementary recovery should be satisfied. 


We now check that each of these algebras satisfy the definition of correctable given above. For $\mathcal{M}_1$, there does exist an $ O_A \in \mathcal{L}(\mathcal{H}_A)$  such that for any $ O_L \in \mathcal{M}$, $ O_L = V^{\dagger}(O_A \otimes I_{\bar{A}})V$. Figure~\ref{fig:calc} demonstrates this calculation.

First applying $V$ to the logical state $\ket{\psi_L} =\ket{\alpha}$, we get the 4 qubit physical state 
\begin{align}
   V\ket{\alpha} =\ket{\psi_\H}= \frac{a}{\sqrt{2}}\ket{0000} +\frac{a}{\sqrt{2}}\ket{0100} +\frac{b}{\sqrt{2}}\ket{1010} + \frac{b}{\sqrt{2}}\ket{1111} .
\end{align}
Then we act $O_A \otimes I_{\bar A}= Z_1 \otimes I_2 \otimes I_{\bar A}$, where the subscripts indicate the qubit the operator is acting on:
\begin{align}
   O_A\ket{\psi_\H} = \frac{a}{\sqrt{2}}\ket{0000} +\frac{a}{\sqrt{2}}\ket{0100} -\frac{b}{\sqrt{2}}\ket{1010} - \frac{b}{\sqrt{2}}\ket{1111}  .
\end{align}
Finally, we act $V^\dagger$ (sending the state backwards through the circuit) to get the output logical state 
\begin{align}
   V^\dagger O_A\ket{\psi_\H} = \frac{a}{2}\ket{0} +\frac{a}{2}\ket{0} -\frac{b}{2}\ket{1} - \frac{b}{2}\ket{1} = Z\ket{\alpha}.
\end{align}
Thus, there exists an $O_A$ such that  $O_L=V^{\dagger}(O_A \otimes I_{\bar{A}})V=Z$. The same can be done with $O_A = I_1 \otimes I_2$ to get $O_L = I$. Since the operators on the logical state are any linear combination of $\{Z,I\}$, then we know any logical operator is correctable from $A_{1,2}$.

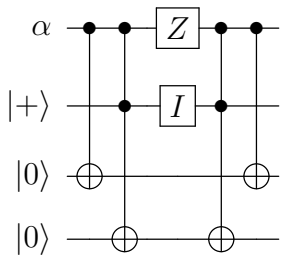
\begin{figure}[t]
\centering

$$  \Qcircuit @C=0.3em @R=1.2em {
   \lstick{\alpha} & \ctrl{2} & \ctrl{1} & \qw & \gate{Z} & \ctrl{1} & \ctrl{2} & \qw \\
    \lstick{\ket{+}} & \qw &  \ctrl{2} & \qw & \gate{I} & \ctrl{2} & \qw & \qw \\
    \lstick{\ket{0}} & \targ & \qw  & \qw & \qw & \qw & \targ & \qw \\
    \lstick{\ket{0}}  & \qw & \targ & \qw &  \qw & \targ & \qw & \qw
}$$

\caption{Applying V, then $O_A \otimes I_{\bar A} = Z_1 \otimes I_2 \otimes I_{\bar A}$, then applying $V^\dagger$}
\label{fig:calc}

\end{figure}

Using the same procedure for $\mathcal{M}_2$, we find there exists an $ O_A \in \mathcal{L}(\mathcal{H}_A)$  such that for any $ O_L \in \mathcal{M}_2$, $ O_L = V^{\dagger}(O_A \otimes I_{\bar{A}})V$. We find if $O_L = Z, \text{ then } O_A = Z_1 \otimes I_3$
  and when $O_L = I$, $O_A = I_1 \otimes I_3$. In addition, when $O_L = X$, $O_A = 2X_1 \otimes X_3$ and when $O_L = Y$, $O_A = 2Y_1 \otimes X_3$. Thus, $\mathcal{M}_2$ is correctable from $A_{1,3}$ with respect to $V$.

  We now check if the lemma defined above holds by checking if these algebras are private from their respective $\bar{A}$ subregions. 

  Our calculations show\footnote{To check this, we exhaustively checked the commutators that appear in \eqref{eq:correct_correctable} for each error operator of the erasure channel for a complete basis of operators in the algebra. Our Mathematica code is available on Github:https://github.com/juliaj25/CR-for-Holographic-Codes.git} that for every $O_{\bar{A}} \in \mathcal{L}(\mathcal{H}_{\bar{A}})$, $ V^{\dagger}(I_{A}\otimes O_{\bar{A}} )V$ commutes with every element of $\M_1$.  However, for every $O_{\bar{A}} \in \mathcal{L}(\mathcal{H}_{\bar{A}})$ it is not the case that $ V^{\dagger}(I_{A}\otimes O_{\bar{A}} )V$ commutes with every element of $\M_2$. For example,
  \begin{align}
      V^{\dagger}(I \otimes (-Z+I) \otimes I \otimes Z)V = Z
  \end{align}
  does not commute with $X\in \M_2$. Thus, $\M_2$ is not private from $\bar{A}_{2,4}$ and the lemma is contradicted.  

Additionally, we have a contradiction with the uniqueness statement in Theorem~\ref{thm:whatalgebra}.  Using the same process as above, we found that $\mathcal{M'}_1$ and $\mathcal{M'}_2$ are correctable from their respective $\bar{A}$ subregions and complementary recovery should be satisfied for both cases. However, let $\mathcal{N}=\text{span}\{Z, I\} $ a subset of $\M_2= \text{span}\{X,Y,Z,I\}$. $\mathcal{N}$ is itself a von Neumann algebra and has the commutant $\mathcal{N}'= \text{span}\{Z,I\}$. We have confirmed that $\text{span}\{Z,I\}$ is correctable from both $A_{1,3}$ and $\bar A_{2,4}$, thus complementary recovery is satisfied. However, we also confirmed $\M_2$ and $\M'_2$ are correctable from $A_{1,3}$ and $\bar A_{2,4}$ and complementary recovery must also be satisfied. Since $\M_2$ is a von Neumann algebra, this contradicts the statement that any von Neumann algebra satisfying complementary recovery is unique.

  \subsection{Modified theorem}

  In order to remedy this contradiction and restate the theorem of existence and uniqueness, we must apply the original operator-algebraic definitions of correctability and privacy given in \cite{kribs2019quantumcomplementarityoperatorstructures} to our specific erasure channel $\mathcal{E} : \mathcal{L}(\mathcal{H}) \rightarrow \mathcal{L}(\mathcal{H})$, discussed in section~\ref{sec:background}.

We must also use a theorem in \cite{kribs2019quantumcomplementarityoperatorstructures} that the incorrect Lemma~\ref{lemma:corrpriv} was adapted from:

\begin{theorem}\label{thm:real}
    [\cite{kribs2019quantumcomplementarityoperatorstructures}, Proposition 2.4] Let $\M$ be a subalgebra of $\mathcal{L}(P\H)$, for some Hilbert space $\H$ and projection $P$. Let $\E$ be a channel on $\H$ with a complementary channel $\E^C$. Then $\M$ is correctable for $\E$ with respect to $P$ if and only if $\M$ is private for $\E^C$ with respect to $P$.
\end{theorem}

When calculating privacy for $\E^C$ we will need the dual of the complementary channel:
\begin{align}\label{eq:compch}
    (\mathcal{E}^{C })^\dagger(\rho) \equiv \sum_{k,l} \text{Tr}(\ket{l}\bra{k} \rho) E_k^* E_l,
\end{align}
  where $\ket{k}$ and $\ket{l}$ are the basis states of the environment space $\H_C$. In our case, $\H_C$ is a four-qubit Hilbert space, since $\E$ has 16 Kraus operators. 
The problem is that Lemma~\ref{lemma:corrpriv} does not follow from Theorem~\ref{thm:real}, because $\E^C$ is not necessarily equal to $\bar \E$, the channel that erases $A$ and tensors on the identity in $A$.

With this in mind, we must now redefine complementary recovery in terms of these quantum channels:

\begin{definition}
    
  If $P$ is a projection, $A$ is a subregion inducing the factorization $\mathcal{H} = \mathcal{H}_A \otimes\mathcal{H}_{\bar{A}} $, and $\M$ is a von Neumann algebra, then $(P, A, \M)$ exhibit complementary recovery if 
  
\begin{itemize}

    \item $\M \text{ is correctable from } \mathcal{E} \text{ with respect to } P$ where $\E: \mathcal{L(H)} \rightarrow \mathcal{L(H)}$ is the channel that erases $\bar{A}$ and tensors on $I_{\bar{A}}$,

       \item $\M' \text{ is correctable from } \bar{\mathcal{E}} \text{ with respect to } P$ where $\bar{\E}: \mathcal{L(H)} \rightarrow \mathcal{L(H)}$ is the channel that erases $A$ and tensors on $I_{A}$.
    
\end{itemize}
\end{definition}

This is a stronger definition of complementary recovery than Definition~\ref{def:wrongcom} because the definition of correctable in Equation (\ref{eq:correct_correctable}) must hold for each combination of Kraus operators. The logical operators need to be recoverable from $A$ for each possible action of the error channel, not just the full erasure of $\bar A$. This is why $\M_2$ is not, in fact, correctable in the operator-algebraic sense despite satisfying Definition~\ref{def:wrong}.

We will now restate and prove Theorem~\ref{thm:whatalgebra} using these new definitions and modified conditions.

\begin{theorem} \textbf{Uniqueness of the von Neumann algebra.} 
  Say $P$ is a projection and $A$ is a subregion. Let $\M:=P(\mathcal{L}(\H_A) \otimes I_{\bar A})P$ be the image of operators projected onto $P\H$. If $\M$ is a von Neumann algebra (that is, it is closed under multiplication) and is correctable from $\E$,
then it is the unique von Neumann algebra satisfying complementary recovery with $P$ and $A$. If it is not, then no von Neumann algebra satisfying complementary recovery exists.
\end{theorem}

\begin{proof}

Let $\M:=P (\mathcal{L}(\mathcal{H}_A)\otimes I_{\bar{A}})P $ be the image of operators on $\mathcal{H}_A$ projected onto $P\mathcal{H}$ that is correctable for the channel $\mathcal{E}: \mathcal{L(H)} \rightarrow \mathcal{L(H)} $ with respect to $P$. To first prove complementary recovery, we assume $\M$ is a von Neumann algebra. The first condition of complementary recovery holds by assumption. Then, by the definition of $\M$, every element of its commutant $\M'$ must commute with $P (\mathcal{L}(\mathcal{H}_A)\otimes I_{\bar{A}})P $. Thus $\M'$ is private from $\bar{\E}^C$ which we show below. 

Then, the Kraus operators of $\bar{\E}$ are 

\begin{align}
    \bar E_{i,j}= \frac{1}{\sqrt{\mathcal{H}_{A}}}( \ket{i}_{A} \otimes I_{\bar{A}})( \bra{j}_{A} \otimes I_{\bar{A}}).
\end{align}

By the definition of privacy, we have

\begin{align}
    P (\bar{\E}^{C})^\dagger (\mathcal{L}(\mathcal{H}_C))P \subseteq S = \{ X \in \mathcal{L}(P\H), [X, \M']=0\}.
\end{align}

Using the definition of the dual complementary channel defined in equation~\ref{eq:compch}, we see this can be written as 

\begin{align}
    P ( \sum_{k,l} (\text{Tr}(\ket{l}\bra{k} \rho)  \bar E_k^* \bar E_l) P = P (\sum_{k,l} \frac{\rho_{k,l}}{|\H_A|}\ket{k}\bra{l}\otimes I_{\bar{A}}) P 
  =  P(\mathcal{L}(\mathcal{H}_A)\otimes I_{\bar{A}})P \subseteq S.
\end{align}

This says $P(\mathcal{L}(\mathcal{H}_A)\otimes I_{\bar{A}})P$ must commute with $\M'$, and by definition of $\M$, this is true. Thus $\M'$ is private for $\bar{\E}^C$. Then by Theorem~\ref{thm:real}, $\M'$ is correctable for $\bar{\E}$ and complementary recovery is satisfied.

Now, to prove uniqueness, let $\mathcal{N} \subsetneq  P (\mathcal{L}(\mathcal{H}_A)\otimes I_{\bar{A}})P $ be any von Neumann algebra that is correctable from $\E$ but not equal to the full set of correctable operators. We assume $(P, A, \mathcal{N})$ obeys complementary recovery and derive a contradiction. By the second condition of complementary recovery, $\mathcal{N'}$ is correctable from $\bar{\E}$ with respect to $P$ and thus is private from $\bar{\E}^C$ with respect to $P$. As shown above, this is 
\begin{align}
    P (\bar{\E}^{C})^\dagger (\mathcal{L}(\mathcal{H}_C))P =  P (\mathcal{L}(\mathcal{H}_A)\otimes I_{\bar{A}})P \subseteq S,
\end{align}
where $S$ is the set of operators that commute with $\mathcal{N}'$. Using the bicommutant theorem, $S=\mathcal{N}''=\mathcal{N}$. We then have 
\begin{align}
    \mathcal{N} \subsetneq  P (\mathcal{L}(\mathcal{H}_A)\otimes I_{\bar{A}})P \subseteq \mathcal{N},
\end{align}
 which is a contradiction.

\end{proof}

The contradiction to the uniqueness result from above is resolved with this revised theorem. Consider again $\mathcal{N} =\text{span}\{Z,I\}$, the subset of $\M_2=\text{span}\{X,Y,Z,I\}$ with commutant $\mathcal{N}' =\text{span}\{Z,I\}$. If complementary recovery is satisfied, $\mathcal{N}$ should be correctable from $\E$ and therefore private from $\E^C$. However, using the definitions of private and complementary channel we find $V^\dagger (\bar{\E}^{C})^\dagger (\mathcal{L}(\mathcal{H}_C))V =  V^\dagger (\mathcal{L}(\mathcal{H}_A)\otimes I_{\bar{A}})V =\text{span}\{X,Y,Z,I\} $ must commute with $\mathcal{N}$ which is not true, so $\mathcal{N}$ is not, in fact, correctable. Thus, complementary recovery is not satisfied.
\section{Other examples of complementary recovery}
\label{sec:examples}

Using the other example codes given in \cite{Pollack_2022}, we were able to analyze\footnote{To accomplish this we wrote Mathematica code, which we have made available on Github:https://github.com/juliaj25/CR-for-Holographic-Codes.git} the structure of the circuits that give rise to the contradiction presented in the section above. Additionally, we were able to find examples of algebras with nonadjacent bipartitions that are still correctable from their subregions. 

\subsection{4-qubit code}
We start with another example of a four-qubit code, this time with a 3-qubit logical state $\ket{\psi_L}= \ket{\alpha i j}$ as shown in Figure~\ref{fig:exd}. It has two choices of bipartitions, the adjacent where $A_{1,2}$ and $\bar A_{3,4}$, and the nonadjacent where $A_{1,3}$ and $\bar A_{2,4}$. We found using Definition~\ref{def:corr} that both algebras are correctable from their respective subregions:
\begin{align}
\M_1 &= \text{span}\{ZXI, ZYI, ZZI, ZII, IXI, IYI, IZI, III\} \\
\M'_1 &= \text{span}\{ZIX, ZIY, ZIZ, ZII, IIX, IIY, IIZ, III\} \\
\M_2 &= \text{span}\{XII, YII, ZII, III\} \\
\M'_2 &= \text{span}\{
IXX, IXY, IXZ, IXI, IYX, IYY, IYZ, IYI, \\
&\quad IZX, IZY, IZZ, IZI, IIX, IIY, IIZ, III \} \notag
\end{align}

For the adjacent case, the logical operations that can be seen with just access to $A$ are those with either $I$ or $Z$ acting on $\alpha$, $X,Y,Z$ or $I$ on $i$ and only $I$ on $j$. This is because $\alpha$ interacts with $\bar A$ so some operations are lost in the erasure $\bar A$, $i$ is in $A$ and has no interaction with $\bar A$, and $j$ is in $\bar A$. 

For the non-adjacent case, we can see all operations on $\alpha$ because it has no interactions with $\bar A$. However, both $i$ and $j$ are in $\bar A$ so no operations on these qubits other than the identity can be seen from $A$.

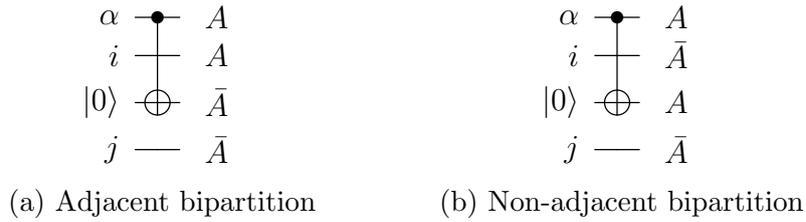
\begin{figure}[t]
\centering

\begin{subfigure}[b]{0.4\textwidth}
        $$\Qcircuit @C=0.3em @R=1.1em {
                \lstick{\alpha} & \ctrl{2} & \qw &  \rstick{A}\\
            \lstick{i} & \qw & \qw &  \rstick{A}\\
            \lstick{\ket{0}} & \targ & \qw &  \rstick{\bar A} \\
            \lstick{j}  & \qw & \qw &  \rstick{\bar A}
}$$
    \caption{Adjacent bipartition}
\end{subfigure}
\begin{subfigure}[b]{0.4\textwidth}
        $$\Qcircuit @C=0.3em @R=1.1em {
                \lstick{\alpha} & \ctrl{2} & \qw &  \rstick{A}\\
            \lstick{i} & \qw & \qw &  \rstick{\bar A}\\
            \lstick{\ket{0}} & \targ & \qw &  \rstick{A} \\
            \lstick{j}  & \qw & \qw &  \rstick{\bar A}
}$$
    \caption{Non-adjacent bipartition}
\end{subfigure}
\caption{Isometry encoding a 3-qubit logical state into a 4-qubit physical state}
\label{fig:exd}
\end{figure}

\subsection{6-qubit code}

 \begin{table}[ht]
    \centering
\begin{tabular}{ | m{5em} | m{1cm}| m{1cm} | m{2cm}|} 
\hline
     & Size of $\M$ & Size of $\M'$ & Correctable \\ 
  \hline
   $A_{1,2,3}$  (adjacent)  & 8 & 8 &yes \\ 
  \hline
    $A_{1,2,6}$ (adjacent)& 8& 8&yes \\ 
  \hline
   $A_{1,5,6}$ (adjacent)  & 8 & 8 &yes\\ 
  \hline
  $A_{1,2,4}$ & 16 & 8 &no \\ 
  \hline
  $A_{1,2,5}$  & 32& 2&yes\\ 
  \hline
  $A_{1,3,5}$ & 8 & 8 &yes\\ 
  \hline
  $A_{1,3,4}$  & 4& 32&no\\ 
  \hline 
   $A_{1,3,6}$& 2& 32 &yes\\ 
  \hline
  $A_{1,4,5}$ & 16 & 8 &no\\ 
  \hline
  $A_{1,4,6}$  & 4 & 32&no\\ 
  \hline
\end{tabular}
 \caption{All bipartitions of the 6 qubit code}
\end{table}

The six-qubit code in Figure~\ref{fig:exg} has a much larger choice of bipartitions. There are ten possible bipartitions of in which there are three physical qubits on each side: three are adjacent and seven are nonadjacent. All three adjacent bipartitions have algebras that are correctable from their respective subregions. Additionally, two of the nonadjacent bipartitions are correctable. This data is given in the table below.

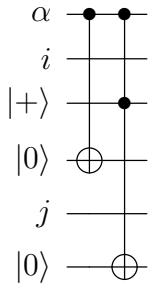
\begin{figure}
$$\Qcircuit @C=0.3em @R=1.3em {
        \lstick{\alpha}  & \ctrl{3} & \ctrl{2} & \qw \\
    \lstick{i}  & \qw & \qw & \qw  \\
    \lstick{\ket{+}} & \qw & \ctrl{3} & \qw  \\
    \lstick{\ket{0}} & \targ & \qw & \qw  \\
    \lstick{j} & \qw  & \qw  & \qw  \\
    \lstick{\ket{0}} & \qw &  \targ & \qw 
}$$
\caption{Isometry encoding a 3-qubit logical state into a 6-qubit physical state}
 \label{fig:exg}  
\end{figure}

For the adjacent bipartitions and $A_{1,3,5}$, the algebras are just those that have $Z$ or $I$ on $\alpha$,  $X,Y,Z$ or $I$ on either $i$ or $j$ depending on which qubit is in $A$, and $I$ on the qubit in $\bar A$. For $A_{1,2,5}$ and $A_{1,3,6}$, both $i$ and $j$ are in either $A$ or $\bar A$, thus the algebras are just $Z$ and $I$ on $\alpha$ and either $X,Y,Z$ or $I$ on both $i$ and $j$, or $I$ on both. The key feature of these bipartitions that makes these algebras correctable is the CNOT is in both $A$ and $\bar A$. When this is not the case, as in the other bipartitions listed, the algebra is no longer correctable. This is because when there is full access to the CNOT, this will allow the creation of $X$ and $Y$ on $\alpha$, while the Toffoli gate  will still create a $Z$ in the other subregion. Thus, we will have an $X$ in the algebra and a $Z$ in its commutant which cannot be true.

\section{Discussion}\label{sec:discussion}

The results presented in this paper better solidify our understanding of which codes exhibit complementary recovery. There are several natural next steps we hope to pursue.

On the one hand, it would be useful to understand more concrete examples of holographic codes at higher qubit number. One possibility might be to glue the explicit constructions presented in \cite{Pollack_2022} and reviewed here into holographic tensor networks \cite{pastawski2015holographic, hayden2016holographic,kohler2019toy,jahn2021holographic,Cao:2021wrb}. Another might be to design codes with more structured logical algebras (and centers). Thus far much of the work has focused exclusively on the area operator, which \cite{Harlow_2017} showed to obey an analogue of the Ryu-Takayanagi (RT) formula \cite{ryu2006holographic}. However, since holographic codes with Clifford encoding circuits have trivial area operators \cite{Cao:2023mzo,Steinberg:2023wll}, the geometric information in stabilizer tensor networks must come entirely from other operators in the center of the code subalgebra.

On the other hand, the notion of geometry provided from a holographic code, as already realized in \cite{Harlow_2017}, is limited compared to that of full holography. In particular, we'd like to have examples of codes where \emph{any} spatial bipartition of the boundary gives complementary recovery. In particular, in the case when the boundary consists of many qubits, we shouldn't need to consider only ``half-and-half'' bipartitions, where the erased region and its complement have the same Hilbert space dimension. However, the algebraic tools we used seem to rely on the fact that the error channel and its complement have the same number of Kraus operators--otherwise, it's not obvious how to interpret the complementary channel as a channel whose input and output are both density matrices in the same Hilbert space.

We also have some direct expectations from holography about how the center of the logical algebra should change as the bipartition changes. For example, if we divide the boundary into two pieces, and then interchange a small region of one piece for a small region of the other piece, we expect the RT surface of the new bipartition to consist of the union of the RT surface of the original bipartition and two much smaller pieces. That suggests that the center of the new logical algebra should be strictly larger than the old logical algebra, and include the original area operator inside it. It would be interesting to understand what features of a code are necessary to capture this behavior.

Given the understanding developed in the previous section the features of encoding circuits required for correctability, it seems there are several additional atomic holographic codes we can construct by modifying the circuits in ways which do not ruin correctability. Most straightforwardly, by swapping CNOT to CZ gates we shift the operator that can be corrected in the logical Hilbert space from $Z$ to $X$. It seems necessary to understand how to build more varied and larger examples of atomic codes to make them more realistic models of holography.

\section*{Acknowledgements}
We thank Jason Crann, Xi Dong, and David Kribs for helpful discussions.

\clearpage


\printbibliography[heading=bibintoc] 

\end{document}